\documentclass[11pt,a4paper]{article}
\usepackage{a4wide}
\usepackage{graphicx}
\usepackage{amsmath,amssymb}
\usepackage{cleveref}
\usepackage{url}
\usepackage{mathdots}
\usepackage{amsthm}
\newtheorem{theorem}{Theorem}
\newtheorem{definition}[theorem]{Definition}

\newtheorem{problem}[theorem]{Problem}
\newcommand{\lzend}{\mathit{LZEnd}}

\newcommand{\Ot}{\tilde{O}}
\begin{document}
\title{Optimal LZ-End Parsing is Hard}
\author{
    Hideo~Bannai$^{1}$ \quad 
    Mitsuru~Funakoshi$^{2,3}$ \quad 
    Kazuhiro~Kurita$^{4}$ \quad \\ 
    Yuto~Nakashima$^{2}$ \quad 
    Kazuhisa~Seto$^{5}$ \quad 
    Takeaki~Uno$^{6}$ \quad \\ 
    \\
    {$^1$ M\&D Data Science Center,} \\
    {Tokyo Medical and Dental University, Tokyo, Japan}\\
    {\texttt{hdbn.dsc@tmd.ac.jp}}\\
    {$^2$ Department of Informatics, Kyushu University, Fukuoka, Japan}\\
    {\texttt{mitsuru.funakoshi@inf.kyushu-u.ac.jp}}\\
    {\texttt{nakashima.yuto.003@m.kyushu-u.ac.jp}}\\
    {$^3$ Japan Society for the Promotion of Science, Tokyo, Japan}\\
    {$^4$ Nagoya University, Nagoya, Japan,} \\
    {\texttt{kurita@i.nagoya-u.ac.jp}}\\
    {$^5$ Faculty of Information Science and Technology,}\\
    {Hokkaido University, Sapporo, Japan}\\
    {\texttt{seto@ist.hokudai.ac.jp}}\\
    {$^6$ National Institute of Informatics, Tokyo, Japan}\\
    {\texttt{uno@nii.jp}}\\
}
\maketitle
\begin{abstract}
    LZ-End is a variant of the well-known Lempel-Ziv parsing family
    such that each phrase of the parsing has a previous occurrence,
    with the additional constraint that the previous occurrence 
    must end at the end of a previous phrase.
    LZ-End was initially proposed as a greedy parsing, 
    where each phrase is determined greedily from left to right,
    as the longest factor that satisfies the above constraint~[Kreft \& Navarro, 2010].
    In this work, we consider an optimal LZ-End parsing that has the minimum number of phrases in such parsings.
    We show that a decision version of computing the optimal LZ-End parsing is NP-complete by showing a reduction from the vertex cover problem.
    Moreover, we give a MAX-SAT formulation for the optimal LZ-End parsing
    adapting an approach for computing various NP-hard repetitiveness measures recently presented by [Bannai et al.,~2022].
    We also consider the approximation ratio of the size of greedy LZ-End parsing to the size of the optimal LZ-End parsing, and give a lower bound of the ratio which asymptotically approaches $2$.
\end{abstract}
\section{Introduction}
In the context of lossless data compression,
various repetitiveness measures --
especially those based on dictionary compression algorithms --
and relations between them have recently received much attention (see the excellent survey by Navarro~\cite{DBLP:journals/corr/abs-2004-02781,repetitiveness_Navarro21a}).
One of the most fundamental and well-known measures is the LZ77 parsing~\cite{LZ77}, in which a string is parsed into $z$ phrases such that each phrase is a single symbol, or the longest which has a previous occurrence.
LZ-End~\cite{lzend_conference,lzend_journal} is a variant of LZ77 parsing with the added constraint that a previous occurrence of the phrase must end at the end of a previous phrase. More formally, the LZ-End parsing is a sequence $q_1, \ldots, q_{z_e}$ of substrings (called phrases) of a given string that can be greedily obtained from left to right: each phrase $q_i$ satisfies (1) $q_i$ is a single symbol which is the leftmost occurrence of the symbol or (2) $q_i$ is the longest prefix of the remaining suffix which is a suffix of $q_1\cdots q_j$ for some $j<i$.
It is known that LZ-End parsing can be computed in linear time~\cite{lzend_ESA17}, and there exists a space-efficient algorithm~\cite{lzend_DCC17}.

While there is no known data structure of $O(z)$ size that provides efficient random access to arbitrary positions in the string,
it was recently shown that $\Ot(1)$ time access could be achieved with $O(z_e)$ space~\cite{lzend_SODA22}.
Furthermore, 
concerning the difference between $z$ and $z_e$,
an upper bound of $z_e = O(z\log^2 (n/z))$ was shown~\cite{lzend_SODA22}, where $n$ is the length of the (uncompressed) string.
On the other hand, there is an obvious bound of 
$z_e= \Omega(z\log n)$ for the unary string, since a previous occurrence of an LZ-End phrase cannot be self-referencing, i.e., overlap with itself, while an LZ77 phrase can.
Notice that $z \leq z_{no} \leq z_e$ holds for any string,
where $z_{no}$ is the number of phrases in the LZ77 parsing that does not allow self-referencing.
A family of strings such that the ratio $z_e/z_{no}$ asymptotically approaches 2 (for large alphabet~\cite{lzend_journal}, for binary alphabet~\cite{lzend_SPIRE21}) is known, 
and it is conjectured that $z_e \leq 2z_{no}$ holds for any strings~\cite{lzend_journal}.

While the phrases in the parsings described above are chosen greedily (i.e., longest), we can consider variants which do not impose such constraint,
e.g.,
in an {\em LZ-End-like} parsing, 
each phrase $q_i$ satisfies (1) $q_i$ is a single symbol which is the leftmost occurrence of the symbol or (2) $q_i$ is a (not necessary longest) prefix of the remaining suffix which is a suffix of $q_1\cdots q_j$ for some $j<i$.
We refer to an LZ-End-like parsing with the smallest number $z_{end}$ 
of phrases, an {\em optimal} LZ-End parsing~\cite{DBLP:journals/corr/abs-2004-02781},
and call the original, the {\em greedy} LZ-End parsing.\footnote{Notice that we do not need the distinction for LZ77, since the greedy LZ77 parsing is also an optimal LZ77-like parsing.}
Thus $z \leq z_{no} \leq z_{end} \leq z_e$ holds.
Interestingly, $z_{end} \leq g$ holds, where $g$ is the size of the smallest context free grammar that derives (only) the string, while a similar relation between $z_e$ and $g$ does not seem to be known~\cite{DBLP:journals/corr/abs-2004-02781}.

This brings us to two natural and important questions about the measure $z_{end}$:
\begin{itemize}
    \item How efficiently can we compute $z_{end}$?
    \item How much smaller can $z_{end}$ be compared to $z_e$?
\end{itemize}

In this work, we answer a part of the above questions. Namely:
\begin{enumerate}
    \item We prove the NP-hardness of computing $z_{end}$.
    \item We present an algorithm for exact computation by MAX-SAT.
    \item We give a lower bound of the maximum value of the ratio $z_e/z_{end}$.
\end{enumerate}

In Section~\ref{sec:hardness}, we give the hardness result.
Our reduction is from the vertex cover problem: finding a minimum set $U$ of vertices such that every edge is incident to some vertex in $U$.
In Section~\ref{sec:maxsat}, we show a MAX-SAT formulation for computing the optimal LZ-End parsing that follows an approach by Bannai et al. that allows computing NP-hard repetitiveness measures using MAX-SAT solvers~\cite{DBLP:conf/esa/Bannai0IKKN22}.
In Section~\ref{sec:lowerbound}, we consider the ratio $z_e/z_{end}$.
We give a family of binary strings such that the ratio asymptotically approaches 2.
Note that we can easily modify this result to a larger alphabet.
Since $(z_e/z_{\mathit{end}}) \leq (z_e/z_{\mathit{no}})$,
the bound is tight, assuming that the conjecture by Kreft and Navarro~\cite{lzend_journal} holds,

\subsection*{Related work}
The LZ77 and LZ78 are original members of the LZ family~\cite{LZ77,LZ78}.
It is well-known that the (greedy) LZ77 parsing produces the optimal version of the parsing~\cite{LZ77-76}.
On the other hand, the NP-hardness of computing the optimal version of the LZ78 variant~\cite{optimal-LZ78}.
This hardness result is also given by a reduction from the vertex cover problem.
However, our construction of the reduction for the LZ-End differs from that for the LZ78 since these parsings have very different structures.

\section{Preliminaries} \label{sec:preliminaries}

\subsection{Strings}
Let $\Sigma$ be an {\em alphabet}.
An element of $\Sigma^*$ is called a {\em string}.
The length of a string $w$ is denoted by $|w|$.
The empty string $\varepsilon$ is the string of length 0.
Let $\Sigma^+$ be the set of non-empty strings,
i.e., $\Sigma^+ = \Sigma^* \setminus \{\varepsilon \}$.
For any strings $x$ and $y$, $x \cdot y$ denotes the concatenation of two strings.
We will sometimes abbreviate ``$\cdot$'' (i.e., $x \cdot y = xy$).
For a string $w = xyz$, $x$, $y$ and $z$ are called
a \emph{prefix}, \emph{substring}, and \emph{suffix} of $w$, respectively.
They are called a \emph{proper prefix}, a \emph{proper substring}, and a \emph{proper suffix} of $w$
if $x \neq w$, $y \neq w$, and $z \neq w$, respectively.
The $i$-th symbol of a string $w$ is denoted by $w[i]$, where $1 \leq i \leq |w|$.
For a string $w$ and two integers $1 \leq i \leq j \leq |w|$,
let $w[i..j]$ denote the substring of $w$ that begins at position $i$ and ends at
position $j$. For convenience, let $w[i..j] = \varepsilon$ when $i > j$.
We will sometimes use $w[i..j)$ to denote $w[i..j-1]$.
For any string $w$, let $w^1 = w$ and let $w^k = ww^{k-1}$ for any integer $k \ge 2$,
i.e., $w^k$ is the $k$-times repetition of $w$.

\subsection{LZ-End parsing}
We give a definition of the LZ-End parsing, which is a variant of the Lempel-Ziv family.
\begin{definition} [(Greedy) LZ-End parsing]
    The \emph{LZ-End parsing} of a string $w$ is the parsing $\lzend(w) = q_{1}, \ldots, q_{z_{end}}$ of $w$
    such that $q_i$ is either a singleton which is the leftmost occurrence of the symbol or the longest prefix of $q_i \cdots q_{z_{end}}$ 
    which occurs as a suffix of $q_{1} \cdots q_{j}$ for some $j<i$, which we call a {\em source} of the phrase.
\end{definition}
This definition, used in~\cite{lzend_SODA22}, is slightly different from the original version~\cite{lzend_conference,lzend_journal} 
where a symbol is added to each phrase.
The results in this paper hold for the original version as well (which we will show in the full version of the paper), but here we use this definition for simplicity.
We refer to each $q_i$ as a \emph{phrase}.
In this paper, we consider a more general version of the LZ-End parsing:
a parsing $q_{1}, \ldots, q_{z_e}$ of a string $w$ 
such that $q_i$ is a (not necessary longest) suffix of $q_{1} \cdots q_{j}$ for some $j<i$.
We call such a parsing with a minimum number $z_e$ of phrases an {\em optimal LZ-End parsing} of $w$.
We give an example of the greedy LZ-End parsing and the optimal LZ-End parsing in Figure~\ref{fig:lzend}.
\begin{figure}[t]
    \begin{center}
    \includegraphics[width=0.8\textwidth]{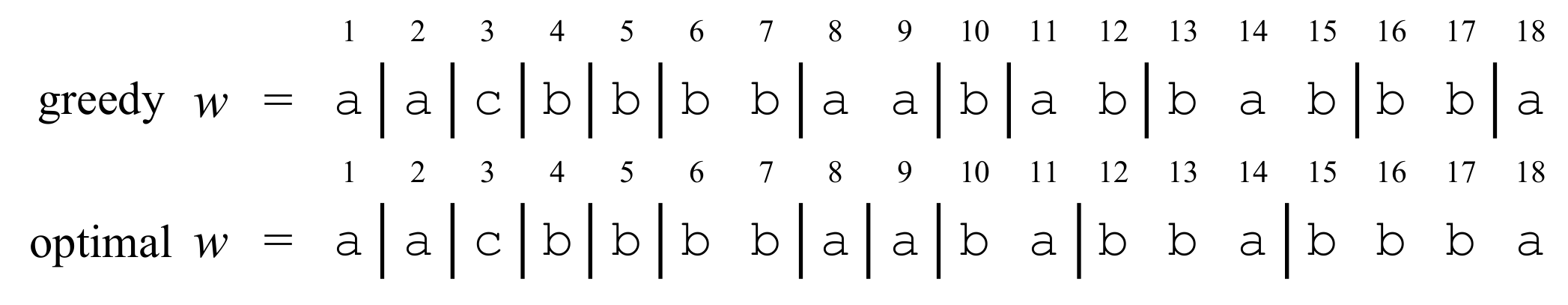}
    \caption{
        Let $w = \mathtt{aacbbbbaababbabbba}$.
        The greedy LZ-End parsing $\lzend(w)$ of $w$ is illustrated in the upper part of the figure.
        For the phrase at position 10, a longer substring $w[10..11] = \mathtt{ba}$ has another previous occurrence at position 7, but there is no phrase that ends at position 8, and any longer substring does not have a previous occurrence
        Therefore, the phrase staring at position 10 is $\mathtt{b}$.
        The lower part of the figure shows an optimal LZ-End parsing (which is smaller than the greedy one) on the same string.
        Each phrase has a previous occurrence that ends at the end of some LZ-End phrase.
        The size of the greedy parsing is 12 and the size of the optimal parsing is 11.
    } 
    \label{fig:lzend}
    \end{center}
\end{figure}

\subsection{Graphs}
Let $G = (V, E)$ be a graph with the set of vertices $V$ and the set of edges $E$.
An edge $e = \{u, v\}$ is called an \emph{incident edge of $u$}.
We denote the set of incident edges of $v$ as $\Gamma_G(v)$.
If there is no fear of confusion, we drop the subscript.
For an edge $e = \{u, v\}$, vertices $u$ and $v$ are the \emph{end points} of $e$.
For a subset of vertices $U \subseteq V$, $U$ is a \emph{vertex cover} if for any $e \in E$, at least one end point of $e$ is contained in $U$.
Let $\tau_G$ be the size of the minimum vertex cover of $G$ (i.e., $\tau_G$ denotes the \emph{vertex cover number} of $G$).
Notice that computing $\tau_G$ is NP-complete~\cite{garey1979computers}.
\section{NP-hardness of computing the optimal LZ-End parsing} \label{sec:hardness}
In this section, we consider the problem of computing the optimal LZ-End parsing of a given string.
A decision version of the problem is given as follows.
\begin{problem}[Decision version of computing the optimal LZ-End parsing~(OptLE)]
    Given a string $w$ and an integer $k$, 
    answer whether there exists an LZ-End parsing of size $k$ or less.
\end{problem}
We show the NP-completeness of $\mathsf{OptLE}$ in the following and present an algorithm for exact computation in the next section.

\begin{theorem}
    $\mathsf{OptLE}$ is NP-complete.
\end{theorem}

\begin{proof}
    We give a reduction from the ``vertex cover problem'' to OptLE.
    Let $G = (V, E)$ be a graph with a set of vertices $V = \{v_1, \ldots, v_n\}$ 
    and a set of edges $E = \{e_1,\ldots, e_m\}$.
    Suppose that an input graph $G$ of the vertex cover problem is connected and $|\Gamma(v)| \geq 2$ for any $v \in V$.
    We identify each vertex $v_i$ as a symbol $v_i$ and each edge $e_i$ as a symbol $e_i$.
    We also introduce the symbol $\$$,
    and a set of symbols that occur uniquely in the string.
    The latter is represented, for simplicity, by the special symbol $\#$,
    i.e., $\#$ represents a different symbol each time it occurs in our description.
    We consider the string $\mathcal{W}_G$ defined by graph $G$ as follows.
    \begin{itemize}
        \item $\mathcal{W}_G = \prod_{i=1}^{n} \mathcal{P}_i \cdot \prod_{j=1}^{m} \mathcal{Q}_j \cdot \prod_{i=1}^{n} \mathcal{R}_i \cdot \prod_{j=1}^{m} \mathcal{S}_j$
        \item $\mathcal{P}_i = v_i^3 \# v_i^2 \$ \# v_i \$^2 \# \cdot \mathcal{X}_i \cdot \mathcal{Y}_i$
        \item $\mathcal{Q}_j = e_j^3 \#$
        \item $\mathcal{R}_i = v_i^4 \$ \prod_{e_j \in \Gamma(v_i)}(e_j^3 v_i^2 \$) \$ \#$
        \item $\mathcal{S}_j = \$ e_j^3 \#$
        \item $\mathcal{X}_i = \prod_{e_j \in \Gamma(v_i)}(v_i \$ e_j \#)$
        \item $\mathcal{Y}_i = \prod_{e_j \in \Gamma(v_i)}(e_j^2 v_i \#)$
    \end{itemize}
    \begin{figure}[t]
        \begin{center}
        \includegraphics[width=0.8\textwidth]{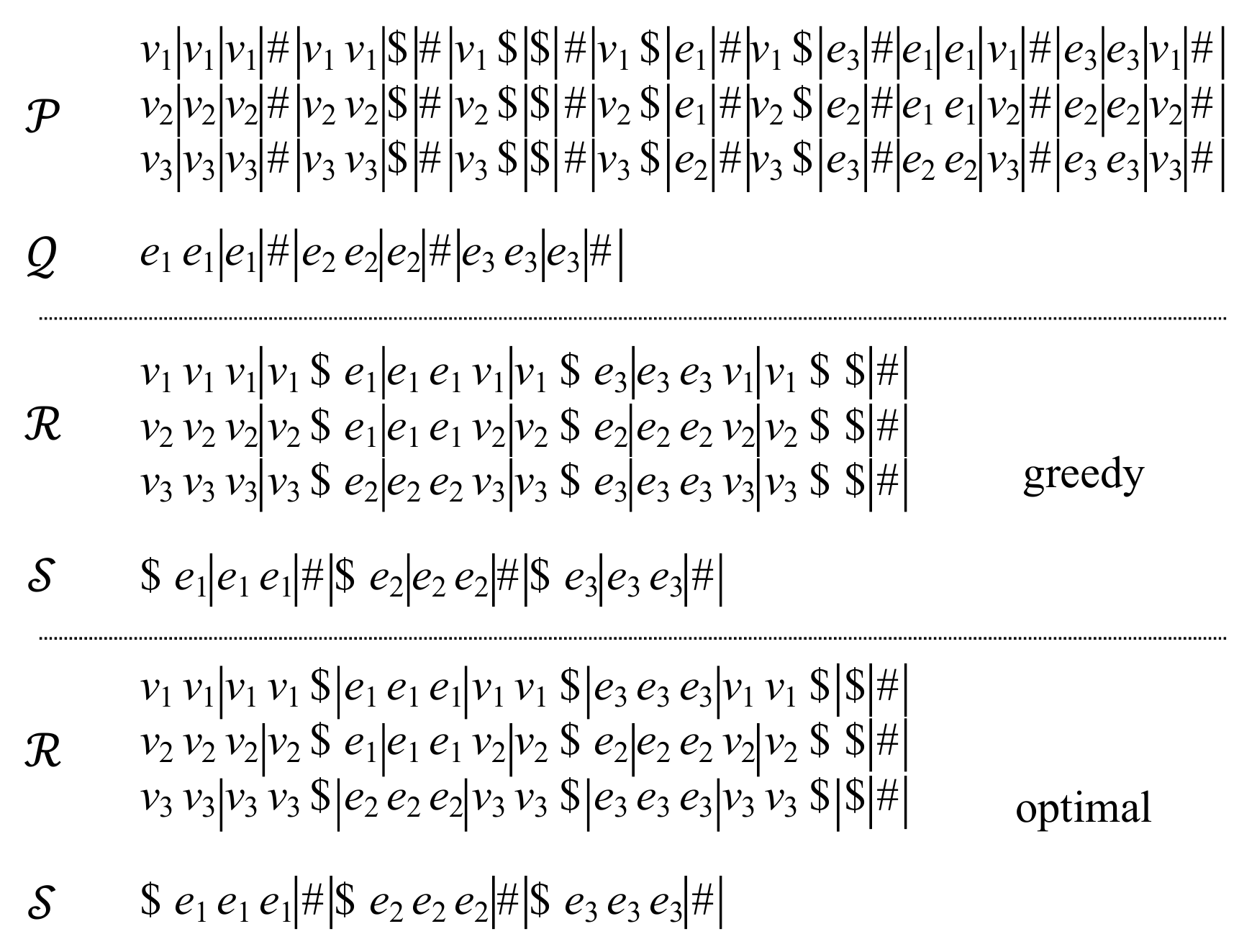}
        \caption{
            Let $G = (V, E)$ be the complete graph of three vertices $v_1, v_2, v_3$ and $e_1 = \{v_1, v_2\}, e_2 = \{v_2, v_3\}, e_3 = \{v_1, v_3\}$.
            $\mathcal{W}_G$ and the greedy parsing and an optimal parsing are illustrated in the figure.
            The first two parts ($\mathcal{P}$ and $\mathcal{Q}$) share the same parsing.
            The last two parts ($\mathcal{R}$ and $\mathcal{S}$) are different.
            The upper part in the figure shows the greedy parsing and the lower part shows an optimal parsing.
            For instance, in the optimal parsing, we can choose $\$e_1^3$ and $\$e_3^3$ as phrases by using non-greedy parsing in $\mathcal{R}_1$.
            In other words, we can reduce two phrases in $\mathcal{S}$-part by adding one phrase in $\mathcal{R}_1$.
            In this example, the optimal parsing represents a vertex cover $\{v_1, v_3\} \subset V$ of $G$ (since $\mathcal{R}_2$ selects the greedy parsing and the others are not).
        }
        \label{fig:gadget}
        \end{center}
    \end{figure}
    An example of this string is illustrated in Figure~\ref{fig:gadget}.
    Note that we use $i$ for representing indices of vertices and $j$ for indices of edges.
    We show that the number of phrases of the optimal parsing of $\mathcal{W}_G$ is less than $13n + 22m + k$
    if and only if the vertex cover number $\tau_G$ is less than $k$.

    First, we observe an optimal LZ-End parsing of $\mathcal{W}_G$.
    Let us consider a parsing of $\prod_{i=1}^{n} \mathcal{P}_i$.
    In this part, the greedy parsing gives $10n + 13m$ phrases.
    In the greedy parsing of $\prod_{i=1}^{n} \mathcal{P}_i$, phrases $v_i^2$ in $v_i^2\$\#$, $v_i\$$ in $v_i\$^2\#, v_i\$e_j\#$, and the second occurrence of $e_j^2$ have length 2, and the other phrases have length 1.
    It is easy to see that this parsing is a smallest possible parsing of $\prod_{i=1}^{n} \mathcal{P}_i$.
    Moreover, other parsings of the same size do not affect the parsing of the rest of the string;
    candidates for a source cannot be increased by selecting any other parsings since the phrases of length 2 are preceded by unique symbols $\#$.
    Hence, we can choose this greedy parsing as a part of an optimal parsing.
    
    In the second part $\prod_{j=1}^{m} \mathcal{Q}_j$, 
    the greedy parsing also gives an optimal parsing which has $3m$ phrases (i.e., each $\mathcal{Q}_j$ is parsed into three phrases since $e_j^2$ occurs in $\mathcal{P}_i$ for some $i$ and $e_j^3$ is unique in $\prod_{i=1}^{n} \mathcal{P}_i \cdot \prod_{j=1}^{m} \mathcal{Q}_j$).
    This parsing is also a smallest possible parsing and does not affect any parsings of the rest of the string.

    The remaining suffix $\prod_{i=1}^{n} \mathcal{R}_i \cdot \prod_{j=1}^{m} \mathcal{S}_j$ is a key of the reduction.
    The key idea is that $\mathcal{S}_j$ represents whether the edge $e_j$ is an incident edge of some vertex in a subset of vertices or not.
    $\$e_j^3$ in $\mathcal{S}_j$ has exactly two previous occurrences in the $\mathcal{R}$-part (since each edge is incident to exactly two vertices).
    Hence $\mathcal{S}_j$ can be parsed into two phrases (i.e., $\$e_j^3, \#$) if and only if $\$e_j^3$ has an occurrence which ends with an LZ-End phrase in the $\mathcal{R}$-part.
    Now we consider the greedy parsing of the $\mathcal{R}_i$-part (let $\Gamma(v_i) = \{ e_{(i, 1)}, \ldots, e_{(i, |\Gamma(v_i)|)}\}$), which is as follows:
    \[
        v_i^3, v_i\$e_{(i, 1)}, e_{(i, 1)}^2v_i, \ldots, v_i\$e_{(i, |\Gamma(v_i)|)}, e_{(i, |\Gamma(v_i)|)}^2v_i, v_i\$^2, \#.
    \]
    The parsing has $2|\Gamma(v_i)|+3$ phrases.
    We claim that this parsing is the smallest possible parsing:
    If the length of every phrases is at most 3, then $2|\Gamma(v_i)|+3$ is the minimum size since the length of $\mathcal{R}_i$ is $6(|\Gamma(v_i)|+1)+1$.
    On the other hand, we can see that substrings of length at least 4 which contain a symbol $v_i$ are unique in the whole string $\mathcal{W}_G$ by the definition.
    Namely, $\$e_j^3$ is the only substring of length at least 4 which is not unique.
    However, any parsing which chooses the length-4 substrings (at the second occurrences) has at least $2|\Gamma(v_i)|+4$ phrases
    because the length-4 substrings do not strictly contain any phrases of the greedy parsing.
    Thus the minimum number of phrases of $\mathcal{R}_i$ is $2|\Gamma(v_i)|+3$ and the above greedy parsing is the only candidate which is the minimum size.
    Notice that phrases of this parsing do not end with $\$e_j^3$.
    Let us consider the other possible parsing of $\mathcal{R}_i$-part as follows:
    \[
        v_i^2, v_i^2\$, e_{(i, 1)}^3, \ldots, v_i^2\$, e_{(i, |\Gamma(v_i)|)}^3, v_i^2\$, \$, \#.
    \]
    This parsing has $2|\Gamma(v_i)|+4$ phrases.
    Notice that this parsing has phrases which end with $\$e_j^3$.
    Thus $\mathcal{S}_j$ can be parsed into two phrases
    if we choose a non-greedy parsing such that there exists a phrase that ends at one of these positions.
    In other words, if we choose such a parsing in the $\mathcal{R}_i$-part, we can reduce at most $|\Gamma(v_i)|$ phrases in the $\mathcal{S}$-part.
    These observations implies that $\mathcal{R}_i$ is parsed into $2|\Gamma(v_i)|+3$ or $2|\Gamma(v_i)|+4$ phrases in any optimal parsing of $\mathcal{W}_G$.

    Let us consider an optimal LZ-End parsing.
    Let $r$ be the number of substrings $\mathcal{R}_i$ which contain $2|\Gamma(v_i)|+4$ phrases,
    and $s$ be the number of substrings $\mathcal{S}_j$ which contain exactly two phrases.
    Then the size of the parsing is 
    \[
        (10n+13m)+(3m)+(2\sum_{i=1}^n |\Gamma (v_i)| +3n+r)+(3m-s) = 13n+23m+r-s.
    \]
    We consider a subset $V'$ of vertices such that $v_i \in V'$ if and only if $\mathcal{R}_i$ is parsed into $2|\Gamma(v_i)|+4$ phrases (i.e., $|V'| = r$), and a subset $E'$ of edges such that $e_j \in E'$ if and only if $\mathcal{S}_i$ is parsed into three phrases (i.e., $|E'| = m-s$).
    If $E' = \emptyset$ (i.e., $s = m$), $V'$ is a vertex cover of $G$.
    Otherwise, $V'$ is not a vertex cover of $G$.
    However we can obtain the vertex cover number by using the parsing.
    Since the parsing is an optimal parsing, we can observe that there is no vertex $v_i$ in $V \setminus V'$ which has two or more incident edges in $E'$ (we can reduce two or more phrases in $\mathcal{S}$-part by adding one phrase in $\mathcal{R}_i$, a contradiction).
    This implies that we can obtain a vertex cover by choosing one vertex in $V \setminus V'$ for each edge in $E \setminus E'$.
    Then there exists an optimal LZ-End parsing of the same size which can directly represent a vertex cover.
    In other words, the vertex cover number is $r+m-s$ if there exists an optimal LZ-End parsing of $13n + 22m + (r+m-s)$ phrases.
    It is clear from the above constructions that there exists an optimal LZ-End parsing of $13n + 22m + k$ phrases if the vertex cover number is $k$.
\end{proof}

\section{MAX-SAT Formulation} \label{sec:maxsat}

An approach for exact computation of various NP-hard repetitiveness measures by formulating them as MAX-SAT instances so that very efficient solvers can be taken advantage of, was shown in~\cite{DBLP:conf/esa/Bannai0IKKN22}.
Here, we show that this approach can be adapted to computing the optimal LZ-End parsing as well.


Let the input string be $T[1..n]$, and for any $i\in[2,n]$, let $M_i = \{ j \mid 1\leq j < i, T[j] = T[i]\}$.
Below, we use $1$ to denote true, and $0$ to denote false.
We introduce the following Boolean variables:
\begin{itemize}
  \item $p_i$ for all $i\in [1,n]$: $p_i = 1$ if and only if position $i$ is a starting position of an LZ-End phrase. Note that $p_1 = 1$.
  \item $c_{i}$ for all $i\in [1,n]$: $c_i = 1$ if and only if position $i$ is the left-most occurrence of symbol $T[i]$.
  \item $r_{i\rightarrow j}$ for all $i\in [2,n]$ and $j\in M_i$: $r_{i\rightarrow j} = 1$ if and only if position $i$ references position $j$ via an LZ-End factor.
\end{itemize}

Notice that the truth values of $c_i$ are all fixed for a given string and is easy to determine.
Furthermore, the left-most occurrence must be beginning of a phrase, so, some values of $p_i$ can also be fixed. For all $i\in [1,n]$:
\begin{align}
  c_i = p_i = 1 & \mbox{ if $i$ is left-most occurrence of $T[i]$}, \\
  c_i = 0       & \mbox{ otherwise. }
\end{align}

MAX-SAT is a variant of SAT, in which there are two types of clauses: hard clauses and soft clauses.
A solution for MAX-SAT is a truth assignment of the variables such that all hard clauses are satisfied (evaluate to $1$), and the number of soft clauses that are satisfied is maximized.

The truth values of $p_i$ define the factors, so in order to minimize the number of factors, 
we define the soft clauses as $\lnot p_i$ for all $i\in[1,n]$.
Below, we give other constraints between the variables that must be satisfied, i.e., hard clauses.

The symbol at any position must either be a left-most occurrence, or it must reference some position to its left.
That is, for any $i\in[1,n]$:
\begin{align}
  c_i + \sum_{j\in M_i} r_{i\rightarrow j} = 1\label{maxsateq:singletonorref}.
\end{align}

In order to ensure that references in the same LZ-End phrase are consistent, we have the following two constraints.
The first ensures that if $i$ references $j$ and the symbols at positions $i-1$ and $j-1$ are different
or does not exist (i.e., $j=1$),
position $i$ and $i-1$ cannot be in the same LZ-End phrase.
For all $i \in[2,n]$ and $j \in M_i$ s.t. $j = 1$ or $T[j-1]\neq T[i-1]$:
\begin{align}
  r_{i\rightarrow j} \implies p_i,
  \label{maxsateq:refadjacent1}
\end{align}
The second ensures that if position $i$ references
position $j$ and $i$ is not a start of an LZ-End phrase, then, position $i-1$ must reference position $j-1$.
For all $i\in [2,n]$ and $j\in M_i\setminus\{1\}$ s.t. $T[j-1]=T[i-1]$:
\begin{align}
  r_{i\rightarrow j} \land \lnot p_i \implies r_{i-1\rightarrow j-1}.
  \label{maxsateq:refadjacent2}
\end{align}

Finally, the following constraints ensure that the reference of each LZ-End phrase must end at an end of a previous LZ-End phrase.
For all $i \in [1,n]$ and $j\in M_i$:
\begin{align}
  \begin{cases}
    r_{i\rightarrow j} \land p_{i+1} \implies p_{j+1} & \mbox{if $i \in [1,n)$} \\
    r_{i\rightarrow j} \implies p_{j+1}               & \mbox{if $i = n$.}
  \end{cases}\label{maxsateq:copyfrom_phraseend}
\end{align}

It is easy to see that the truth assignments that are derived from any LZ-End parsing will satisfy the above constraints.

We now show that any truth assignment that satisfies the above constraints will represent a valid LZ-End parsing.
The truth values for $p_i$ implies a parsing where each phrase starts at a position $i$ if and only if $p_i = 1$.
Constraint~(\ref{maxsateq:singletonorref}) ensures that each position is either a singleton, or references a unique previous position.
Thus, it remains to show that the referencing of each position of a given factor is consistent (adjacent positions reference adjacent positions) and ends at a previous phrase end.

For any position $i$ such that $c_i = 0$,
let $j\in M_i$ be the unique value such that $r_{i\rightarrow j}=1$.
We can see that any such position $i$ that is not at the beginning of a phrase (i.e., $p_i = 0$) will reference a position consistent with the reference of position $i-1$:
If $j = 1$ or $T[j-1]\neq T[i-1]$, then Constraint~(\ref{maxsateq:refadjacent1}) would imply $r_{i\rightarrow j} = 0$. Thus, we have $j > 1$ and $T[j-1]= T[i-1]$, and from Constraint~(\ref{maxsateq:refadjacent2}),
we have that $r_{i-1\rightarrow j-1}$, and the referencing inside a factor is consistent.
Finally, from Constraint~(\ref{maxsateq:copyfrom_phraseend}),
the last reference in a phrase always points to an end of a previous LZ-End phrase.

The MAX-SAT instance contains $O(n^2)$ variables,
and the total size of the CNF is $O(n^2)$:
$O(n)$ clauses of $O(n)$ size (Constraint~(\ref{maxsateq:singletonorref}) using linear size encodings of cardinality constraints, e.g.~\cite{DBLP:conf/cp/Sinz05}),
and $O(n^2)$ clauses of size $O(1)$
(the soft clauses, and Constraints~(\ref{maxsateq:refadjacent1}),~(\ref{maxsateq:refadjacent2}),~(\ref{maxsateq:copyfrom_phraseend})).

We note that it is not difficult to obtain a MAX-SAT formulation for the original definition of LZ-End by minor modifications.

\section{Approximation ratio of greedy parsing to optimal parsing} \label{sec:lowerbound}

In this section, we consider an approximation ratio of the size $z_e$ of the greedy LZ-End parsing to the size $z_{end}$ of the optimal LZ-End parsing.
Here, we give a lower bound of the ratio.

\begin{theorem} \label{thm:lowerbound}
    There exists a family of binary strings 
    such that the ratio $z_e/z_{end}$ asymptotically approaches $2$.
\end{theorem}
\begin{proof}
    Let $K = \sum_{i=1}^{k} 2^i$ for any positive integer $k \geq 1$.
    The following binary string $w_k$ over an alphabet $\{a, b\}$ gives the lower bound:
    \[
        w_k = aa \cdot \prod_{i=1}^{k}(a^{2^i}) \cdot b^4 \cdot \prod_{i=1}^{K}(a^i b^3).
    \]
    It is easy to see that $K$ is the length of the substring $\prod_{i=1}^{k}(a^{2^i})$.
    First, we show the greedy parsing of $w_k$.
    Let $W_0 = aa \cdot \prod_{i=1}^{k}(a^{2^i}) \cdot b^4$ (i.e., a prefix of $w_k$) and
    $W_j = W_{j-1} \cdot a^j b^3$ for any $1 \leq j \leq K$.
    Notice that $W_K = w_k$.
    We show that 
    \begin{equation} \label{eq:lower-greedy}
        \lzend(W_j) = \lzend(W_{j-1}), a^jb^2, b
    \end{equation}
    by induction on $j$.
    Initially, we consider the greedy parsing of $W_0$.
    The greedy parsing of the first run (i.e., maximal substring with a unique symbol) is $a, a, a^2, \ldots, a^{2^k}$ of size $k+2$.
    The second run is parsed into three phrases $b, b, b^2$.
    Thus 
    \[
        \lzend(W_0) = a, a, a^2, \ldots, a^{2^k}, b, b, b^2.
    \]
    Moreover, $\lzend(W_1) = \lzend(W_0), ab^2, b$ since $W_1 = W_0 \cdot ab^3$ holds.
    Hence Equation~\ref{eq:lower-greedy} holds for $j=1$.
    Suppose that Equation~\ref{eq:lower-greedy} holds for any $j \leq p$ for some integer $p \geq 1$.
    We show that Equation~\ref{eq:lower-greedy} holds for $j = p+1$.
    Assume that there exists a phrase $x$ of $\lzend(W_{p+1})$ which begins in $W_p$ and ends in a new suffix $a^{p+1}b^3$ of $W_{p+1}$.
    By the induction hypothesis, phrases of $\lzend(W_p)$ which end with $a$ are only in the first $a$'s run.
    This implies that $x$ cannot end with $a$ and $x$ can be written as $x = x'ba^{p+1}b^{\ell}$ for some prefix $x'$ of $x$ and some positive integer $\ell$.
    However, $a^{p+1}$ only occurs in the first $a$'s run.
    Thus $\lzend(W_{p+1})$ cannot have such a phrase $x$, namely $\lzend(W_{p+1}) = \lzend(W_{p}), S$ for some factorization of $S$.
    It is easy to see that the remaining suffix $a^{p+1}b^3$ of $W_{p+1}$ is parsed into $a^{p+1}b^2, b$.
    Hence Equation~\ref{eq:lower-greedy} holds for $j=p+1$, and it also holds for any $j$.
    Notice that $|\lzend(w_k)| = 2K+k+5$ holds.

    Finally, we give a smaller parsing of $w_k$.
    We consider the same parsing for the first run and a different parsing for the second run as $b, b, b, b$.
    In the greedy parsing, $a^jb^3$ cannot be a phrase since the only previous occurrence does not have an LZ-End phrase.
    We can use a substring $a^jb^3$ as a new phrase of $W_j$ (see also Figure~\ref{fig:lowerbound}).
    Thus there exists an LZ-End parsing 
    \[
        a, a, a^2, \ldots, a^{2^k}, b, b, b, b, a^1b^3, \ldots, a^Kb^3.
    \]
    The size of the parsing is $K+k+6$.

    Therefore the ratio $z_e/z_{end}$ asymptotically approaches $2$ for this family of strings.
\end{proof}

\begin{figure}[t]
    \begin{center}
    \includegraphics[width=\textwidth]{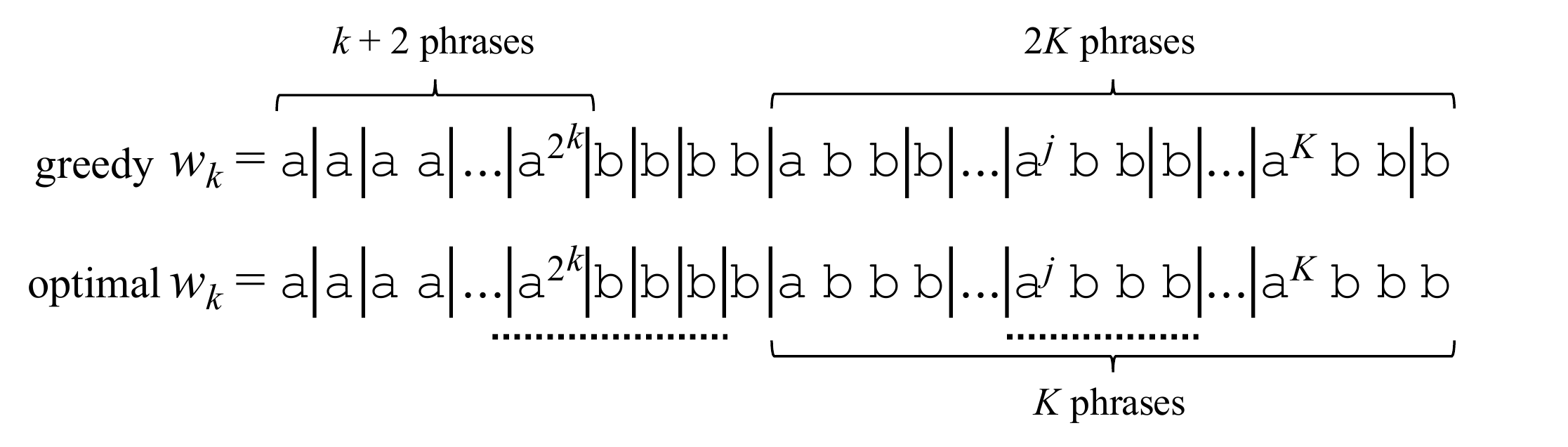}
    \caption{
        Illustration for two variants of LZ-End parsings of a string $w_{k}$ (Theorem~\ref{thm:lowerbound}).
        In the optimal parsing, we can choose $a^jb^3$ (dotted lines) as a phrase for each $j~(1 \leq j \leq K)$ by adding a single letter phrase $b$.
    }
    \label{fig:lowerbound}
    \end{center}
\end{figure}

Note that this family of strings also gives a lower bound of the ratio $z_{e}/z_{no}$ since $(z_e/z_{\mathit{end}}) \leq (z_e/z_{\mathit{no}})$ holds.

%
    
\section{Conclusions}
In this paper, we first studied the optimal version of the LZ-End variant.
We showed the NP-completeness of the decision version of computing the optimal LZ-End parsing and presented an approach for exact computation of the optimal LZ-End by formulating as MAX-SAT instances.
We also gave a lower bound of the possible gap (as the ratio) between the greedy LZ-End and the optimal LZ-End.
Finally, we note possible future work in the following.
\begin{itemize}
    \item Our reduction from the vertex cover problem uses a polynomially large alphabet.
    How can we construct a reduction with a small alphabet?
    \item The most interesting remaining problem is an upper bound of the ratio discussed in Section~\ref{sec:lowerbound}.
    We conjecture that there exists a constant upper bound (i.e., $z_e/z_{end} \leq c$ for any strings where $c$ is a constant).
    This implies that the greedy parsing gives a constant-approximation of the optimal parsing.
    On the other hand, if there exists a family of strings which gives $c>2$ or non-constant ratio, then the conjecture $z_e \leq 2z_{no}$ does not stand.
\end{itemize}
\section*{Acknowledgments}
We would like to thank Dominik K\"oppl for discussion.
\clearpage
\bibliographystyle{abbrv}
\bibliography{ref}
\end{document}